\documentclass[11pt]{article}
\usepackage{fullpage}
\usepackage{graphicx}
\usepackage{amssymb}
\usepackage{amsmath}
\usepackage{amsthm}

\usepackage{enumitem}

\newtheorem{theorem}{Theorem}[section]
\newtheorem{lemma}[theorem]{Lemma}

\newtheorem{invariant}[theorem]{Invariant}

\theoremstyle{definition}

\newcommand{\ignore}[1]{}
             
\newcommand{\hcm}[1][1]{\hspace*{#1 cm}}
\newcommand{\rb}[2]{\raisebox{#1 mm}[0mm][0mm]{#2}}
\newcommand{\istrut}[2][0]{\rule[- #1 mm]{0mm}{#1 mm}\rule{0mm}{#2 mm}}
\newcommand{\zero}[1]{\makebox[0mm][l]{$#1$}}

\newcommand{\paren}[1]{{\mathopen{}\left( #1 \right)\mathclose{}}}

\newcommand{\angbrack}[1]{\left< #1 \right>}

\newcommand{\ceil}[1]{\lceil #1 \rceil}
\newcommand{\floor}[1]{\lfloor #1 \rfloor}
\newcommand{\f}[2]{\frac{#1}{#2}}

\newcommand{\poly}{\operatorname{poly}}
\newcommand{\bottom}{\perp}

\newcommand{\MASS}{\operatorname{mass}}
\newcommand{\Euler}{\operatorname{Euler}}
\newcommand{\JOIN}{\mbox{\sc Join}}
\newcommand{\SPLIT}{\mbox{\sc Split}}
\newcommand{\INSERTEDGE}{\mbox{\sc Insert}}
\newcommand{\DELETEEDGE}{\mbox{\sc Delete}}
\newcommand{\CONN}{\mbox{\sc Conn?}}
\newcommand{\LIST}{\mbox{\sc List}}

\newcommand{\REPLACEMENTEDGE}{\mbox{\sc ReplacementEdge}}

\newcommand{\SCINSERT}{\mbox{\sc SCInsert}}
\newcommand{\SCDELETE}{\mbox{\sc SCDelete}}
\newcommand{\SCJOIN}{\mbox{\sc SCJoin}}
\newcommand{\SCSPLIT}{\mbox{\sc SCSplit}}
\newcommand{\UPDATEADJ}{\mbox{\sc UpdateAdj}}
\newcommand{\ADJQUERY}{\mbox{\sc AdjQuery}}
\newcommand{\MEMQUERY}{\mbox{\sc MembQuery}}

\newcommand{\ID}{\operatorname{ID}}
\newcommand{\MEM}{\operatorname{Memb}}
\newcommand{\ADJ}{\operatorname{SupAdj}}
\newcommand{\ChADJ}{\operatorname{ChAdj}}

\newcommand{\Patrascu}{P\v{a}tra\c{s}cu}

\begin{document}
\title{Faster Worst Case Deterministic Dynamic Connectivity\thanks{This work is supported in part by the Danish National Research Foundation grant DNRF84 through the Center for Massive Data Algorithmics (MADALGO).
S. Pettie is supported by NSF grants CCF-1217338, CNS-1318294, and CCF-1514383.
M. Thorup's research is partly supported by
Advanced Grant DFF-0602-02499B from the Danish Council for Independent
Research under the Sapere Aude research career programme.}}

\author{Casper Kejlberg-Rasmussen\thanks{MADALGO, Aarhus University}\\ 
\and
Tsvi Kopelowitz\thanks{University of Michigan.  Email: kopelot@gmail.com}\\ 
\and
Seth Pettie\thanks{University of Michigan.  Email: pettie@umich.edu}\\ 
\and
Mikkel Thorup\thanks{University of Copenhagen.  Email: mthorup@di.ku.dk}\\ 
}

\date{} 
\maketitle

\begin{abstract}
We present a deterministic dynamic connectivity data structure for undirected graphs 
with worst case update time $O\paren{\sqrt{\frac{n(\log\log n)^2}{\log n}}}$ and constant query time.
This improves on the previous best deterministic worst case algorithm of 
Frederickson ({\em STOC}, 1983) and Eppstein Galil, Italiano, and Nissenzweig ({\em J. ACM}, 1997), which had update time $O(\sqrt{n})$.  
All other algorithms for dynamic connectivity are either randomized (Monte Carlo)
or have only amortized performance guarantees.
\end{abstract}

\section{Introduction}

{\em Dynamic Connectivity} is perhaps the single most fundamental unsolved problem in the area of dynamic graph algorithms.
The problem is simply to maintain a dynamic undirected graph $G=(V,E)$ subject to edge updates and connectivity queries:

\begin{description}[noitemsep]
\item[$\INSERTEDGE(u,v)$ : ] Set $E \leftarrow E \cup \{(u,v)\}$.
\item[$\DELETEEDGE(u,v)$ : ] Set $E \leftarrow E \setminus \{(u,v)\}$.
\item[$\CONN(u,v)$ : ]  Determine whether $u$ and $v$ are in the same connected component in $G$.
\end{description}

Over thirty year ago Frederickson~\cite{Frederickson85} introduced {\em topology trees} and {\em 2-dimensional} topology trees,
which gave the first non-trivial solution to the problem.  Each edge insertion/deletion is handled in $O(\sqrt{m})$ time and each 
query is handled in $O(1)$ time.  Here $m$ is the current number of edges and $n$ the number of vertices. 
On sparse graphs (where $m=O(n)$) Frederickson's data structure has seen no unqualified improvements or simplifications.  
However, when the graph is dense Frederickson's data structure can be sped up using the general sparsification method of
Eppstein, Galil, Italiano, and Nissenzweig~\cite{EppsteinGIN97}.  Using simple sparsification~\cite{EppsteinGIN92} the update
time becomes $O(\sqrt{n}\log(m/n))$ and using more sophisticated sparsification~\cite{EppsteinGIN97} 
the running time becomes $O(\sqrt{n})$.  This last bound has not been improved in twenty years.

Most research on the dynamic connectivity problem has settled for {\em amortized} update time guarantees.
Following~\cite{HenzingerK99,HenzingerT97}, Holm et al.~\cite{HolmLT01} gave a very simple deterministic 
algorithm with amortized update time $O(\log^2 n)$ and query time $O(\log n/\log\log n)$.\footnote{Any connectivity structure that maintains (internally) a spanning forest can have query time $O(\log_{t_u/\log n} n)$ if the update time is $t_u = \Omega(\log n)$.}
However, in the worst case Holm et al.'s~\cite{HolmLT01} 
update takes $\Omega(m)$ time, the same for computing a spanning tree from scratch!
Recently Wulff-Nilsen~\cite{Wulff-Nilsen13} improved the update time of~\cite{HolmLT01} to $O(\log^2 n/\log\log n)$.
Using Las Vegas randomization, Thorup~\cite{Thorup00} gave a dynamic connectivity structure with an 
$O(\log n(\log\log n)^3)$ amortized update time.  In other words, the algorithm answers all connectivity queries correctly
but the amortized update time holds with high probability.

In a major breakthrough Kapron, King, and Mountjoy~\cite{KapronKM13} used Monte Carlo randomization 
to achieve a worst case update time of $O(\log^5 n)$.  However, this algorithm has three notable drawbacks.
The first is that it is susceptible to undetected false negatives: $\CONN(u,v)$ may report that $u,v$ are disconnected when they are, in fact, connected.
The second is that even when $\CONN(u,v)$ (correctly) reports that $u,v$ are connected, it is forbidden from exhibiting a connectivity witness, i.e., a spanning forest in which $u,v$ are joined by a path.
The Kapron et al.~\cite{KapronKM13} algorithm \underline{does} maintain such a spanning forest internally, but if this witness were made public, 
a very simple attack could force the algorithm to answer connectivity queries incorrectly.
Lastly, the algorithm uses $\Omega(n\log^2 n)$ space, which for sparse graphs is superlinear in $m$.
Very recently Gibb et al.~\cite{GibbKKT15} reduced the update time of~\cite{KapronKM13} to $O(\log^4 n)$.

On special graph classes, dynamic connectivity can often be handled more efficiently.  
For example, Sleator and Tarjan~\cite{SleatorT83} maintain a dynamic set of trees in 
$O(\log n)$ worst-case update time subject to $O(\log n)$ time connectivity queries.  (See also~\cite{HenzingerK99,AlstrupHLT05,AcarBHVW04,TarjanW05}.)
Connectivity in dynamic planar graphs can be reduced to the dynamic tree problem~\cite{EppsteinITTWY92,EppsteinITTWY93},
and therefore solved in $O(\log n)$ time per operation.  
The cell probe lower bounds of \Patrascu{} and Demaine~\cite{PatrascuD06} show 
that Sleator and Tarjan's bounds are optimal in the sense that {\em some} operation must take $\Omega(\log n)$ time.
Superlogarithmic updates can be used to get modestly sublogarithmic queries, 
but \Patrascu{} and Thorup~\cite{PatrascuT11} prove the reverse is not possible.
In particular, any dynamic connectivity algorithm with $o(\log n)$ update time
has $n^{1-o(1)}$ query time.  
Refer to Table~\ref{table:priorwork} in the appendix for a history of upper and lower bounds for dynamic connectivity.

\subsection{New Results}

In this paper we return to the classical model of deterministic worst case complexity.  We give a new dynamic connectivity 
structure with worst case update time on the order of
\[
\min\left\{ \sqrt{\frac{m(\log\log n)^2}{\log n}}, \; \; \;  \sqrt{\frac{m\log^5 w}{w}} \right\},
\]
where $w = \Omega(\log n)$ is the word size.\footnote{Our algorithms use the standard repertoire of $AC^0$ operations: left and right shifts, bitwise operations on words, additions and comparisons.  They do not assume unit-time multiplication.}  
These are the first improvements to Frederickson's 2D-topology trees~\cite{Frederickson85} in over 30 years.
Using the sparsification reduction of Eppstein et al.~\cite{EppsteinGIN97} the running time expressions can be made to depend on `$n$' rather than `$m$', so we obtain $O(\sqrt{\frac{n(\log\log n)^2}{\log n}})$ bounds (or faster) for all graph densities.

Compared to the amortized algorithms~\cite{HolmLT01,Thorup00,Wulff-Nilsen13}, 
ours is better suited to {\em online} applications that demand a bound on the latency of {\em every} operation.\footnote{Amortized
data structures are most useful when employed by offline algorithms that do not care about individual operation times. 
The canonical example is the use of amortized Fibonacci heaps~\cite{FT87} to implement Dijkstra's algorithm~\cite{Dij59}.}
Compared to the Monte Carlo algorithms~\cite{KapronKM13,GibbKKT15},
ours is attractive in applications that demand linear space, zero probability of error, and a public witness of connectivity.

The modest speedup obtained by our algorithm over~\cite{Frederickson85,EppsteinGIN97} 
comes from word-level parallelism, which is a widely used in both theory and practice.  
However, just because the underlying machine can operate on $w=\Omega(\log n)$ bits at once
does not mean that $\poly(w)$-factor speedups come easily or automatically.  
The true contribution of this work is in reorganizing the representation of the graph 
{\em so that} word-level parallelism becomes a viable technique.  As a happy byproduct, we 
develop an approach to worst case dynamic connectivity that is conceptually simpler 
than Frederickson's (2-dimensional) topology trees.  

\paragraph{Organization.} 
In Section~\ref{sect:highlevel} we describe our high level approach, without getting into low-level data structural details.
Section~\ref{sect:newstructure} gives a relatively simple instantiation of the high-level approach 
with update time $O(\sqrt{n}/w^{1/4})$, which is slightly slower than our claimed result.
In Section~\ref{sect:speedup} we describe the modifications needed to achieve the claimed bounds.

\section{The High Level Algorithm}\label{sect:highlevel}

The algorithm maintains a spanning tree of each connected component of the graph as a {\em witness} of connectivity.
Each such witness tree $T$ is represented as an Euler tour $\Euler(T)$.\footnote{Henzinger and King~\cite{HenzingerK99} were the first to use Euler tours to represent dynamic trees.  G. Italiano (personal communication) observed that Euler tours could be used in lieu of Frederickson's topology trees to obtain an $O(\sqrt{m})$-time dynamic connectivity structure.}
$\Euler(T)$ is the sequence of vertices encountered in some Euler tour around $T$,
as if each undirected edge were replaced by two oriented edges.
It has length precisely $2(|V(T)|-1)$ if $|V(T)|\ge 2$ (the last vertex is excluded from the list, which is necessarily the same as the first) 
or length 1 if $|V(T)|=1$.
Vertices may appear in $\Euler(T)$ several times.
We designate one copy of each vertex the {\em principle copy}, which is responsible for all edges incident to the vertex.
Each vertex in the graph maintains a pointer to its principle copy.  
Each $T$-edge $(u,v)$ maintains two pointers to the (possibly non-principle) copies of $u$ and $v$ that precede the oriented occurrences of $(u,v)$ and $(v,u)$ in $\Euler(T)$, respectively.  Note that cyclic rotations of $\Euler(T)$ are also valid Euler tours; if $\Euler(T) = (u,\ldots,v)$ the last element of the list is associated with the tree edge $(v,u)$.

When an edge $(u,v)$ that connects distinct witness trees $T_0$ and $T_1$ is inserted, $(u,v)$ becomes a tree edge and 
we need to construct $\Euler(T_0 \cup \{(u,v)\} \cup T_1)$ from $\Euler(T_0)$ and $\Euler(T_1)$.  In the reverse situation,
if a tree edge $(u,v)$ is deleted from $T=T_0\cup \{(u,v)\} \cup T_1$ we first construct $\Euler(T_0)$ and $\Euler(T_1)$ from $\Euler(T)$,
then look for a {\em replacement edge}, $(\hat{u},\hat{v})$ with $\hat{u}\in V(T_0)$ and $\hat{v}\in V(T_1)$. If a replacement is found we
construct
$\Euler(T_0\cup\{(\hat{u},\hat{v})\}\cup T_1)$ from $\Euler(T_0)$ and $\Euler(T_1)$.
Lemma~\ref{lem:surgical} establishes the nearly obvious
fact that the new Euler tours can be obtained from the old Euler tours using 
$O(1)$ of the following {\em surgical operations}: splitting and concatenating lists of vertices, 
and creating and destroying singleton lists containing non-principle copies of vertices.

\begin{lemma}\label{lem:surgical}
If $T = T_0 \cup \{(u,v)\} \cup T_1$ and $(u,v)$ is deleted, $\Euler(T_0)$ and $\Euler(T_1)$ can be constructed from $\Euler(T)$ with $O(1)$ surgical operations.  In the opposite direction, from $\Euler(T_0)$ and $\Euler(T_1)$ we can construct $\Euler(T_0\cup \{(u,v)\} \cup T_1)$ with $O(1)$ surgical operations.  It takes $O(1)$ time to determine which surgical operations to perform.
\end{lemma}

\begin{proof}
Recall that cyclic shifts of Euler tours are valid Euler tours.
Suppose without loss of generality that $\Euler(T) = (P_0,u,v,P_1,v,u,P_2)$ where $P_0,P_1,$ and $P_2$ are sequences of vertices.  (Note that Euler tours never contain immediate repetitions.  If $P_1$ is empty then $\Euler(T)$ would be just $(P_0,u,v,u,P_2)$; if both $P_0$ and $P_2$ are empty then $\Euler(T) = (u,v,P_1,v)$.)
Then we obtain $\Euler(T_0) = (P_0,u,P_2)$ and $\Euler(T_1) = (v,P_1)$ with $O(1)$ surgical operations, which includes the destruction of non-principle copies of $u$ and $v$; at least one of the two copies must be non-principle.  We could also set $\Euler(T_1) = (P_1,v)$, which would be more economical if the $v$ following $P_1$ in $\Euler(T)$ were the principle copy.

In the reverse direction, write $\Euler(T_0) = (P_0,u,P_1)$ and $\Euler(T_1) = (P_2,v,P_3)$, where the labeled occurrences
are the principle copies of $u$ and $v$.
Then $\Euler(T_0 \cup \{(u,v)\} \cup T_1) = (P_0,u,v,P_3,P_2,v,u,P_1)$, where the new copies of $u$ and $v$ are clearly non-principle copies.  
If $P_2$ and $P_3$ were empty (or $P_0$ and $P_1$ were empty) then we would not need to add a non-principle copy of $v$ (or a non-principle copy of $u$.)
\end{proof}

Define $\hat m$ to be an upper bound on $m$, the number of edges.  The update time of our data structure will be a function of $\hat{m}$. 
The sparsification method of~\cite{EppsteinGIN97}
creates instances in which $\hat m$ is known to be linear in the number of vertices.  

\subsection{A Dynamic List Data Structure}

We have reduced dynamic connectivity in graphs to implementing several simple operations on dynamic lists.
We will maintain a pair $(\mathcal{L},E)$, where $\mathcal{L}$ is a set of lists (containing principle
and non-principle copies of vertices) and 
$E$ is the dynamic set of edges joining principle copies of vertices.
In addition to the creation and destruction of single element lists we must support the following
primitive operations.

\begin{description}[noitemsep]
\item[$\LIST(x)$ : ] Return the list in $\mathcal{L}$ containing element $x$.
\item[$\JOIN(L_0,L_1)$ : ]  Set $\mathcal{L} \leftarrow \mathcal{L} \setminus \{L_0,L_1\} \cup \{L_0L_1\}$,
that is, replace $L_0$ and $L_1$ with their concatenation $L_0L_1$.
\item[$\SPLIT(x)$ : ] Let $L=L_0 L_1 \in \mathcal{L}$, where $x$ is the last element of $L_0$.
Set $\mathcal{L} \leftarrow \mathcal{L} \setminus \{L\} \cup \{L_0,L_1\}$.
\item[$\REPLACEMENTEDGE(L_0,L_1)$ : ] Return any edge joining elements in $L_0$ and $L_1$.
\end{description}

Our implementations of these operations will only be efficient if, after each $\INSERTEDGE$ or $\DELETEEDGE$ operation,
there are no edges connecting distinct lists.  That is,
the $\REPLACEMENTEDGE$ operation is only employed by $\DELETEEDGE$ when deleting
a tree edge in order to restore Invariant~\ref{inv:EulerTour}.

\begin{invariant}\label{inv:EulerTour}
Each list $\mathcal{L}$ corresponds to the Euler tour of a spanning tree of some connected component.
\end{invariant}

The dynamic connectivity operations are implemented as follows.
To answer a $\CONN(u,v)$ query we simply check whether $\LIST(u)=\LIST(v)$.
To insert an edge $(u,v)$ we do $\INSERTEDGE(u,v)$, and if $\LIST(u)\neq\LIST(v)$ then make $(u,v)$ a tree edge
and perform suitable $\SPLIT$s and $\JOIN$s to merge the Euler tours $\LIST(u)$ and $\LIST(v)$.
To delete an edge $(u,v)$ we do $\DELETEEDGE(u,v)$, and if $(u,v)$ is a tree edge in $T = T_0\cup \{(u,v)\} \cup T_1$, 
perform suitable $\SPLIT$s and $\JOIN$s to create $\Euler(T_0)$ and $\Euler(T_1)$ from $\Euler(T)$.   
At this point Invariant~\ref{inv:EulerTour} may be violated as there could be an edge joining $T_0$ and $T_1$.  
We call $\REPLACEMENTEDGE(\Euler(T_0),\Euler(T_1))$ and if it finds an edge, say $(\hat{u},\hat{v})$, 
we perform more $\SPLIT$s and $\JOIN$s to form $\Euler(T_0 \cup \{(\hat{u},\hat{v})\}\cup T_1)$.

Henzinger and King~\cite{HenzingerK99} observed that most off-the-shelf balanced binary search trees can support
$\SPLIT$, $\JOIN$, and other operations in logarithmic time.  
However, they provide no direct support for the $\REPLACEMENTEDGE$ operation, 
which is critical for the dynamic connectivity application.

\section{A New Dynamic Connectivity Structure}\label{sect:newstructure}

\subsection{Chunks and Superchunks}

In order to simplify the maintenance of Invariant~\ref{inv:mass}, stated below,
we shall assume that the maximum degree never exceeds $K$, where 
$K \approx \sqrt{\hat{m}/\poly(w)}$ is a parameter of the algorithm.  
Refer to Appendix~\ref{sect:nodegreebound} for a discussion of clean ways to remove
this assumption.

If $L'$ is a sublist of a list $L\in \mathcal{L}$, define 
$\MASS(L')$ to be the number of edges incident 
to elements of $L'$, counting an edge twice if both endpoints are in $L'$.
The sum of list masses, $\sum_{L\in\mathcal{L}} \MASS(L)$, is clearly at most $2\hat m$, 
where $\hat{m}$ is the fixed upper bound on the number of edges.
We maintain a partition of each list $L \in \mathcal{L}$ into {\em chunks}
satisfying Invariant~\ref{inv:mass}.

\begin{invariant}\label{inv:mass}
Let $L\in \mathcal{L}$ be an Euler tour.  If $\MASS(L) < K$ then $L$ consists of a single chunk.
Otherwise $L = C_0C_1\cdots C_{p-1}$ is partitioned into $\Theta(\MASS(L)/K)$ chunks
such that $\MASS(C_l) \in [K,3K]$ for all $l\in [p]$.
\end{invariant}

The chunks are partitioned into contiguous sequences of $\Theta(h)$ superchunks according to Invariant~\ref{inv:superchunk}.
For the time being define $h = 2\floor{\sqrt{w}/2}$, where $w$ is the word size.

\begin{invariant}\label{inv:superchunk}
A list in $\mathcal{L}$ having fewer than $h/2$ chunks
forms a single superchunk with ID $\bottom$.
A list in $\mathcal{L}$ with at least $h/2$ chunks is partitioned 
into superchunks, each consisting of between $h/2$ and $h-1$ consecutive chunks.
Each such superchunk has a unique ID in $[J]$, where $J = 4\hat{m}/(Kh)$.  
(IDs are completely arbitrary.  They do not encode any information about the order of superchunks within a list.)
\end{invariant}

Call an Euler tour list {\em short} if it consists of fewer than $h/2$ chunks.  
We shall assume that no lists are ever short, as this simplifies the description of the data structure and its
analysis.  In particular, all superchunks have proper IDs in $[J]$.  
In Section~\ref{sect:shortlists} we sketch the uninteresting complications introduced by 
$\bottom$ IDs and short lists.

\subsection{Word Operations}\label{sect:word-operations}

When $h\le \floor{\sqrt{w}}$, Invariant~\ref{inv:superchunk} implies
that we can store a matrix $A \in \{0,1\}^{h\times h}$ in one word
that represents the adjacency between the chunks within two superchunks $i$ and $j$.
This matrix will always be represented in row-major order; rows and columns are indexed by $[h] = \{0,\ldots,h-1\}$.
In this format it is straightforward to insert a new all-zero row above a specified row $k$ (and destroy row $h-1$)
by shifting the old rows $k,\ldots,h-2$ down by one.
It is also easy to copy an interval of rows from one matrix to another.
Lemma~\ref{lem:mask} shows that the corresponding operations on {\em columns} can also be effected in $O(1)$ time
with a fixed mask $\mu$ precomputable in $O(\log w)$ time.

\begin{lemma}\label{lem:mask}
Let $h=2\floor{\sqrt{w}/2}$ and let $\mu$ be the word $(1^h0^h)^{h/2}$.
Given $\mu$ we can in $O(1)$ time 
copy/paste any interval of columns from/to a matrix $A\in\{0,1\}^{h\times h}$, represented in row-major order.
\end{lemma}

\newcommand{\CRIGHTSHIFT}{\texttt{>>}}
\newcommand{\CLEFTSHIFT}{\texttt{<<}}
\newcommand{\CNOT}{\texttt{$\sim$}}
\newcommand{\COR}{\texttt{|}}
\newcommand{\CAND}{\texttt{\&}}

\begin{proof}
Recall that the rows and columns are indexed by integers in $[h]=\{0,\ldots,h-1\}$.
We first describe how to build a mask $\nu_k$ for columns $k,\ldots,h-1$ then illustrate how it is used to copy/paste intervals of columns.
In C notation,\footnote{The operations $\CAND,\COR,$ and $\CNOT$ are bit-wise AND, OR, and NOT; $\CLEFTSHIFT$ and $\CRIGHTSHIFT$ are left and right shift.} the word $\nu_k' = (\mu \,\CRIGHTSHIFT\, k) \,\CAND\, \mu$ is a mask for the intersection of the even rows and columns 
$k,\ldots,h-1$, so $\nu_k = \nu_k' \,\COR\, (\nu_k' \,\CRIGHTSHIFT\, h)$ is a mask for columns $k$ through $h-1$.

To insert an all-zero column before column $k$ of $A$ (and delete column $h-1$)
we first copy columns $k,\ldots,h-2$ to $A' = A \,\CAND\, (\nu_{k+1} \,\CLEFTSHIFT\, 1)$ then
set $A = (A \,\CAND\, (\CNOT \nu_{k})) \,\COR\, (A' \,\CRIGHTSHIFT\, 1)$.  Other operations can be effected in $O(1)$ time
with copying/pasting intervals of columns, e.g., splitting an array into two about a designated column, or merging two arrays having at most 
$h$ columns together.
\end{proof}

\subsection{Adjacency Data Structures}

In order to facilitate the efficient implementation of $\REPLACEMENTEDGE$ we maintain
an $O(\hat{m}/K) \times O(\hat{m}/K)$ adjacency matrix between chunks, and a $J \times J$ adjacency matrix between 
superchunks.  However, in order to allow for efficient dynamic updates it is important that these matrices be represented
in a non-standard format described below.  The data structure maintains the following information.

\begin{itemize}[noitemsep]
\item Each list element maintains a pointer to the chunk containing it.  
Each chunk maintains a pointer to the superchunk containing it, as well as an index in $[h]$ indicating its position within the 
superchunk.  
Each superchunk maintains its ID in $[J] \cup \{\bottom\}$ and a pointer to the list containing it.

\item $\ChADJ$ is a $J\times J$ array of $h^2$-bit words ($h^2 \le w$) indexed by superchunk IDs.
The entry $\ChADJ(i,j)$ is interpreted as an $h\times h$ 0-1 matrix
that keeps the adjacency information between all pairs of chunks in superchunk $i$ and superchunk $j$.  (It may be that $i=j$.)
In particular,
$\ChADJ(i,j)(k,l) = 1$ iff there is an edge with endpoints in the $k$th chunk of superchunk $i$ and the $l$th chunk of 
superchunk $j$, so $\ChADJ(i,j) = 0$ (i.e., the all-zero matrix) if no edge joins superchunks $i$ and $j$.  
The matrix $\ChADJ(i,j)$ is stored in row-major order.

\item Let $S$ be a superchunk with $ID(S)=\bottom$.  By Invariants~\ref{inv:EulerTour} and \ref{inv:superchunk},
$S$ is not incident to any other superchunks and has fewer than $h/2$ chunks.  
We maintain a single word $\ChADJ_S$ which stores the adjacency matrix of the chunks within $S$.

\item For each superchunk with ID $i\in[J]$ we keep length-$J$ bit-vectors $\ADJ_i$ and $\MEM_i$, where
\begin{align*}
\ADJ_i(j) &= 1 \mbox{ if $\ChADJ(i,j) \neq 0$ and 0 otherwise, whereas }\\
\MEM_i(j) &= 1 \mbox{ if $j=i$ and 0 otherwise.}
\end{align*}
These vectors are packed into $\ceil{J/w}$ machine words, so scanning one takes $O(\ceil{J/w})$ time.

\item We maintain a {\em list-sum} data structure that allows us to take the bit-wise OR of the 
$\ADJ_i$ vectors or $\MEM_i$ vectors, over all superchunks in an Euler tour.
It is responsible for maintaining the $\{\ADJ_i,\MEM_i\}$ vectors described above and supports the following operations.
At all times the superchunks are partitioned into a set $\mathcal{S}$ of disjoint lists of superchunks.
Each $S\in\mathcal{S}$ (a list of superchunks) is associated with an $L\in\mathcal{L}$ (an Euler tour), 
though short lists in $\mathcal{L}$ have no need for a corresponding list in $\mathcal{S}$.
\begin{description}
\item[$\SCINSERT(i)$ : ] Retrieve an unused ID, say $i'$, and allocate a new superchunk with ID $i'$ and all-zero vector $\ADJ_{i'}$.
Insert superchunk $i'$ immediately after superchunk $i$ in $i$'s list in $\mathcal{S}$.  
If no $i$ is given, create a new list in $\mathcal{S}$ consisting of superchunk $i'$.
\item[$\SCDELETE(i)$ : ] Delete superchunk $i$ from its list and make ID $i$ unused.
\item[$\SCJOIN(S_0,S_1)$ : ] Replace superchunk lists $S_0,S_1\in \mathcal{S}$ with their concatenation $S_0S_1$.
\item[$\SCSPLIT(i)$ : ] Let $S = S_0S_1 \in \mathcal{S}$ and $i$ be the last superchunk in $S_0$.  Replace $S_0S_1$ with two lists $S_0,S_1$.
\item[$\UPDATEADJ(i,x\in\{0,1\}^J)$ : ] Set $\ADJ_i \leftarrow x$ and update $\ADJ_j(i) \leftarrow x(j)$ for all $j\neq i$.
\item[$\ADJQUERY(S)$ : ] Return the vector $\alpha \in \{0,1\}^J$ where
\[
\alpha(j) = \bigvee_{i \in S} \ADJ_i(j)
\]
The index $i$ ranges over the IDs of all superchunks in $S$.

\item[$\MEMQUERY(S)$ : ] Return the vector $\beta \in \{0,1\}^J$, where
\[
\beta(j) = \bigvee_{i\in S} \MEM_i(j)
\]
\end{description}
\end{itemize}

We use the following implementation of the {\em list-sum} data structure.  Each list of superchunks is maintained as any  $O(1)$-degree search tree that supports logarithmic time inserts, deletes, splits, and joins.  Each leaf is a superchunk
that stores its two bit-vectors.  Each internal node $z$ keeps two bit-vectors, $\ADJ^z$ and $\MEM^z$, which are 
the bit-wise OR of their leaf descendants' respective bit-vectors.  Because length-$J$ bit-vectors can be updated in $O(\ceil{J/w})$ time, all ``logarithmic time'' operations on the tree actually take $O(\log J \cdot J/w)$ time.
The $\UPDATEADJ(i,x)$ operation takes $O(\log J \cdot J/w)$ time to update superchunk $i$ and its $O(\log J)$ ancestors.
We then need to update the $i$th bit of potentially every other node in the tree, in $O(J)$ time.
Since $w=\Omega(\log n)=\Omega(\log J)$ the cost per $\UPDATEADJ$ is $O(J)$.
The answer to an $\ADJQUERY(S)$ or $\MEMQUERY(S)$ is stored at the root of the tree on $S$.

\subsection{Creating and Destroying (Super)Chunks}

There are essentially two causes for the creation and destruction of (super)chunks.  
The first is in response to a $\SPLIT$ operation that forces a (super)chunk to be broken up.  (The $\SPLIT$ may itself be instigated by the insertion or deletion of an edge.)  The second is to restore Invariants~\ref{inv:mass} and \ref{inv:superchunk} after a 
$\JOIN$ or $\INSERTEDGE$ or $\DELETEEDGE$ operation.  
In this section we consider the problem of updating the adjacency data structures after four types of operations: 
(i) splitting a chunk in two, keeping both chunks in the same superchunk,
(ii) merging two adjacent chunks in the same superchunk, 
(iii) splitting a superchunk along a chunk boundary,
and (iv) merging adjacent superchunks.
Once we have bounds on (i)--(iv), implementing the higher-level operations in the stated bounds is relatively straightforward.
Note that (i)--(iv) may temporarily violate Invariants~\ref{inv:mass} and \ref{inv:superchunk}.

\paragraph{Splitting Chunks}
Suppose we want to split the $k$th chunk of
superchunk $i$ into two pieces, both of which will (at least temporarily) stay within 
superchunk $i$.\footnote{Remember that `$k$' refers to the actual position of the chunk within its superchunk 
whereas `$i$' is an arbitrary ID that does not relate to its position within the list.}
We first zero-out all bits of $\ChADJ(i,\star)(k,\star)$ and $\ChADJ(\star,i)(\star,k)$ in $O(J)$ time.
For each $j$ we need to insert an all-zero row below row $k$ in $\ChADJ(i,j)$ and an all-zero column after 
column $k$ of $\ChADJ(j,i)$.  This can be done in $O(1)$ time for each $j$, or $O(J)$ in total; see Lemma~\ref{lem:mask}.

In $O(K)$ time we scan the edges incident to the new chunks $k$ and $k+1$
and update the corresponding bits
in $\ChADJ(i,\star)(k',\star)$ and $\ChADJ(\star,i)(\star,k')$, for $k'\in\{k,k+1\}$.

\paragraph{Merging Adjacent Chunks}
In order to merge chunks $k$ and $k+1$ of superchunk $i$ we need to replace
row $k$ of $\ChADJ(i,j)$, for all $j$, with the bit-wise OR of rows $k$ and $k+1$ of $\ChADJ(i,j)$,
zero out row $k+1$, then scoot rows $k+2,\cdots$ back one row.  A similar transformation 
is performed on columns $k$ and $k+1$ of $\ChADJ(j,i)$, 
which takes $O(1)$ time per $j$, by Lemma~\ref{lem:mask}.
In total the time is $O(J)$, independent of $K$.

\paragraph{Splitting Superchunks}
Suppose we want to split superchunk $i$ after its $k$th chunk.  
We first call $\SCINSERT(i)$, which allocates an empty superchunk with ID $i'$ and inserts $i'$ after 
$i$ in its superchunk list in $\mathcal{S}$.
In $O(J)$ time we transfer rows $k+1,\ldots,h-1$ from $\ChADJ(i,j)$ to $\ChADJ(i',j)$
and transfer columns $k+1,\ldots,h-1$ from $\ChADJ(j,i)$ to $\ChADJ(j,i')$.
By Lemma~\ref{lem:mask} this takes $O(1)$ time per $j$.

At this point $\ChADJ$ is up-to-date but the list-sum data structure and $\{\ADJ_j\}$ bit-vectors are not.
We update $\ADJ_i,\ADJ_{i'}$
with calls to $\UPDATEADJ(i,x)$ and $\UPDATEADJ(i',x')$.  Using $\ChADJ$, 
each bit of $x$ and $x'$ can be generated in constant time.  This takes $O(J)$ time.

\paragraph{Merging Superchunks}
Let the two adjacent 
superchunks have IDs $i$ and $i'$.  It is guaranteed that they will be merged only if they contain at most $h$ chunks together.
In $O(J)$ time we transfer the non-zero rows of $\ChADJ(i',j)$ to $\ChADJ(i,j)$ and transfer the non-zero columns of $\ChADJ(j,i')$ to $\ChADJ(j,i)$.
A call to $\SCDELETE(i')$ deletes superchunk $i'$ from its list in $\mathcal{S}$ and retires ID $i'$.
We then call $\UPDATEADJ(i,x)$ with
the new incidence vector $x$.  In this case we can generate $x$ in $O(J/w)$ time since it is merely the bit-wise OR
of the old vectors $\ADJ_i$ and $\ADJ_{i'}$, with bit $i'$ set to zero.  Updating the list-sum data structure takes $O(J)$ time.

\subsection{Joining and Splitting Lists}\label{sect:joiningsplittinglists}

Once we have routines for splitting and merging adjacent (super)chunks, implementing $\JOIN$ and $\SPLIT$ on lists in $\mathcal{L}$ 
is much easier.  The goal is to restore Invariant~\ref{inv:mass} governing chunk masses and Invariant~\ref{inv:superchunk} on the number of chunks per superchunk.

\paragraph{Performing $\JOIN(L_0,L_1)$.} Write $L_0 = C_{0},\ldots,C_{p-1}$ and $L_1 = D_{0},\ldots,D_{q-1}$ as a list of chunks.
If both $L_0$ and $L_1$ are not short then they have corresponding superchunk lists $S_0,S_1\in\mathcal{S}$.
Call $\SCJOIN(S_0,S_1)$ to join $S_0,S_1$ in $\mathcal{S}$, in $O(J)$ time.

\paragraph{Performing $\SPLIT(x)$.} Suppose $x$ is contained in chunk $C_l$ of 
$L = C_0\cdots C_{l-1} C_l C_{l+1}\cdots C_{p-1}$.
We split $C_l$ into two chunks $C_l' C_l''$, and split the superchunk containing $C_l$ along this line.
Let $S$ be the superchunk list corresponding to $L$ and $i$ be the ID of the superchunk ending at $C_l'$.  
We split $S$ using a call to $\SCSPLIT(i)$, which corresponds to splitting $L$ into
$L_0 = C_0\cdots C_{l-1}C_l'$ and $L_1 = C_l''C_{l+1}\cdots C_{p-1}$.
At this point $C_l'$ or $C_l''$ may violate Invariant~\ref{inv:mass} if $\MASS(C_l') < K$ or $\MASS(C_l'')<K$.
Furthermore, Invariant~\ref{inv:superchunk} may be violated
if the number of chunks in the superchunks containing $C_l'$ and $C_l''$ is too small.
We first correct Invariant~\ref{inv:mass} by possibly merging and resplitting $C_{l-1}C_l'$ and $C_l''C_{l+1}$ along new
boundaries.  If the superchunk containing $C_l'$ has fewer than $h/2$ chunks, it and the superchunk to its left
have strictly between $h/2$ and $3h/2$ chunks together, and so can be merged (and possibly resplit)
into one or two superchunks satisfying Invariant~\ref{inv:superchunk}. 
The same method can correct a violation of $C_l''$'s superchunk.  This takes $O(K + J)$ time.

\paragraph{Performing $\REPLACEMENTEDGE(L_0,L_1)$}
The list-sum data structure makes implementing the $\REPLACEMENTEDGE(L_0,L_1)$ operation easy.
Let $S_0$ and $S_1$ be the superchunk lists corresponding to Euler tours $L_0$ and $L_1$.
We compute the vectors $\alpha \leftarrow \ADJQUERY(S_0)$ and $\beta \leftarrow \MEMQUERY(S_1)$
and their bit-wise AND $\alpha\wedge \beta$ with a linear scan of both vectors.
If $\alpha\wedge\beta$ is the all-zero vector then there is no edge between $L_0$ and $L_1$.
On the other hand, if $(\alpha\wedge\beta)(j) = 1$, then $j$ must be the ID of a 
superchunk in $S_1$ that is incident to {\em some} superchunk in $S_0$.  
To determine {\em which} superchunk in $S_0$
we walk down from the root of $S_0$'s list-sum tree to a leaf, say with ID $i$, 
in each step moving to a child $z$ of the current node 
for which $\ADJ^z(j) =1$.  Once $i$ and $j$ are known we retrieve any 1-bit in the matrix 
$\ChADJ(i,j)$, say at position $(k,l)$, indicating
that the $k$th chunk of superchunk $i$ and the $l$th chunk of superchunk $j$ are adjacent.  
We scan all its adjacent edges in $O(K)$ time and retrieve an edge 
joining $L_0$ and $L_1$.  The total time is $O(J/w + \log J + K) = O(J/w+K)$.

\paragraph{Performing $\INSERTEDGE(u,v)$}
If $\LIST(u)\neq\LIST(v)$, first perform $O(1)$ $\SPLIT$s and $\JOIN$s to restore the Euler tour Invariant~\ref{inv:EulerTour}.
Now $u$ and $v$ are in the same list in $\mathcal{L}$.  
Let $i,j$ be the IDs of the superchunks containing the principle copies of $u$ and $v$
and let $k,l$ be the positions of $u$ and $v$'s chunks within their respective superchunks.
We set $\ChADJ(i,j)(k,l) \leftarrow 1$.  If $\ChADJ(i,j)$ was formerly the all-zero matrix, 
we call $\UPDATEADJ(i,x)$ to update superchunk $i$'s adjacency information with the correct vector $x$.\footnote{Since $x$ only differs from the former 
$\ADJ_i$ at position $\ADJ_i(j)$, this update to the list-sum tree 
takes just $O(\log J)$ time since it only affects ancestors of leaves $i$ and $j$.}
Inserting one edge changes the mass of the chunks containing $u$ and $v$, which could violate Invariant~\ref{inv:mass}.
Invariants~\ref{inv:mass} and \ref{inv:superchunk} are restored by splitting/merging $O(1)$ chunks and superchunks.

\paragraph{Performing $\DELETEEDGE(u,v)$} 
Compute $i,j,k,l$ as defined above, in $O(1)$ time.
After we delete $(u,v)$ the correct value of the bit $\ChADJ(i,j)(k,l)$ is uncertain.
We scan chunk $k$ or superchunk $i$ in $O(K)$ time, looking for an edge connected to chunk $l$ of superchunk $j$.
If we do {\em not} find such an edge we set $\ChADJ(i,j)(k,l) \leftarrow 0$, 
and if that makes $\ChADJ(i,j) = 0$ (the all-zero matrix), 
we call $\UPDATEADJ(i,x)$, where $x$ is the new adjacency vector of superchunk $i$;
it only differs from the former $\ADJ_i$ at position $j$.

If $(u,v)$ is a tree edge in $T = T_0 \cup \{(u,v)\} \cup T_1$
we perform $\SPLIT$s and $\JOIN$s to replace $\Euler(T)$ with $\Euler(T_0),\Euler(T_1)$, which may 
violate Invariant~\ref{inv:EulerTour} if there is a replacement edge between $T_0$ and $T_1$.  
We call $\REPLACEMENTEDGE(\Euler(T_0),\Euler(T_1))$ to find a replacement edge. 
If one is found, say $(\hat{u}, \hat{v})$, we 
form $\Euler(T_0 \cup \{(\hat{u},\hat{v})\} \cup T_1)$ with a constant number of $\SPLIT$s and $\JOIN$s.

\subsection{Running Time Analysis}\label{sect:runningtime}

Each operation ultimately involves splitting/merging $O(1)$ chunks, superchunks, and lists,
which takes time $O(K + J + \log n \cdot J/w) = O(K + J) = O(K + \hat{m}/(K\sqrt{w}))$.
We balance the terms by setting $K = \sqrt{\frac{\hat{m}}{\sqrt{w}}}$ so the running time is $O(K)$.

By the sparsification transformation of Eppstein, Galil, Italiano, and Nissenzweig~\cite{EppsteinGIN97} this implies
an update time of $O\paren{\frac{\sqrt{n}}{w^{1/4}}}$.
Each instance of dynamic connectivity created by~\cite{EppsteinGIN97} has a fixed set of vertices, say of 
size $\hat{n}$, and a fixed upper bound $\hat{m} = O(\hat{n})$ on the number of edges.

\section{Speeding Up the Algorithm}\label{sect:speedup}

Observe that there are $\Theta((\hat{m}/(Kh))^2)$ matrices $(\ChADJ(i,j))$ but only $\hat{m}$ edges, 
so for $K = \sqrt{\hat{m}/h}$, the average $h\times h$ matrix has $O(h)$ 1s.
Thus, storing each such matrix verbatim, using $h^2$ bits, is information theoretically inefficient on average.
By storing only the locations of the 1s in each matrix we can represent each matrix in $O(h\log h)$ bits on average
and thereby hope to solve dynamic connectivity faster with a larger `$h$' parameter.

\paragraph{The Encoding.} In this encoding we index rows and columns by indices in $\{1,\ldots,h\}$ rather than $[h]$.
Let $m_{i,j}=m_{j,i}$ be the number of 1s in $\ChADJ(i,j)$.
We encode $\ChADJ(i,j)$ by listing its 1 positions in $O(m_{i,j}\log h / w)$ lightly packed words.
Each word is partitioned into {\em fields} of $1 + 2\ceil{\log(h+1)}$ bits: each field consists of a control bit (normally 0),
a row index, and a column index.  Each word is between half-full and full, the fields in use being packed
contiguously in the word.  This invariant allows us to insert a new field after a given field in $O(1)$ time.
We list the 1s of {\em either} $\ChADJ(i,j)$ {\em or} $\ChADJ(j,i) = \ChADJ(i,j)^{\top}$ in row-major order, 
with a bit indicating which of the two representations is used.

\paragraph{Fast Operations.} 
Given $\ChADJ(i,j)$ in row-major order, we can determine if $\ChADJ(i,j)(k,l)=1$ 
in $O(\log((m_{i,j}\log h)/w))=O(\log h)$ time, as follows.
By doing a binary search over the first field in each word we can determine which word (if any)
has a field containing $\angbrack{k,l}$: the binary encoding of $(k,l)$.  If we add $2^{2\ceil{\log(h+1)}} - \angbrack{k,l}$ to each
field in the word, the control bits for all fields that are equal to or greater than $\angbrack{k,l}$ will be flipped to 1.
Similarly, if we set all control bits to 1 and subtract $\angbrack{k,l}+1$ from each field, the control bits of
fields that are equal to or less than $\angbrack{k,l}$ will be flipped to 0.  Thus, we can single out the control bit
for an occurrence of $\angbrack{k,l}$ (if any) with $O(1)$ bit-wise operations.  If $\angbrack{k,l}$ is not present,
the control bits reveal the field in the word after which it could be inserted, if we need to set $\ChADJ(i,j)(k,l) \leftarrow 1$.

In the same time bound we can also identify the positions of the first and last 1s in row $k$.  Thus,
we can perform the following operations on $\ChADJ(i,j)$ in $O((m_{i,j}\log h)/w)$ time: 
setting a row to zero, incrementing/decrementing the row-index of some interval of rows,
or copying an interval of rows.

The operations sketched above are only efficient if $\ChADJ(i,j)$ is in row-major order.
If we have $\ChADJ(i,j)^{\top}$ in row-major order we can effect a transpose by (1) swapping the row and column
indices in each field using masks and shifts, and (2) sorting the fields.  In general, sorting $x$ words of $O(w/\log h)$ fields
takes $O(x(\log^2(w/\log h) + \log x \log(w/\log h)))$ time using Albers and Hagerup's implementation~\cite{AlbersH97} of 
Batcher's bitonic mergesort~\cite{CLRS09}.\footnote{Albers and Hagerup also require that the fields to be sorted begin with control bits.}
We sort each word in $O(\log^2 (w/\log h))$ time, resulting in $x$ sorted lists, then iteratively merge the two shortest lists
until one list remains.  Merging two lists containing $y$ words takes 
$O(y\log(w/\log h))$ time: we can merge the next $w/\log h$ fields of each list in $O(\log(w/\log h))$ 
time~\cite{AlbersH97} and output at least $w/\log h$ items to the merged list.

Alternatively, if $w = \log n$ we can sort and merge lists of $\epsilon \log n/\log h$ fields in unit time using table lookup
to precomputed tables of size $O(n^\epsilon)$.  In this case sorting $x$ packed words takes $O(x\log x)$ time.

\paragraph{Splitting and Joining.} The cost of splitting and joining (super)chunks is now slightly more expensive.
When handling superchunk $i$ (or any chunk within it) we first put each $\ChADJ(i,j)$ in row-major order, in
$\sum_{j=1}^{J} O(\ceil{\f{m_{i,j}\log h}{w}}\log^2 h) = O(J\log^2 h + (Kh/w)\log^3 h)$ since, 
by Invariants~\ref{inv:mass} and \ref{inv:superchunk}, $\sum_j m_{i,j} = O(Kh)$.
Once the relevant superchunks are in the correct format, splitting or joining $O(1)$ (super)chunks takes
$O(K\log h + J + (Kh/w)\log h)$ time.  Since $J = O(\hat{m}/(Kh))$, the overall update time is
\[
O\paren{K\log h + \f{\hat{m}\log^2 h}{Kh} + \f{Kh\log^3 h}{w}}
\]
Setting $h = w$ and $K = \sqrt{\frac{\hat{m}}{w\log w}}$, the overall time is
$O(\sqrt{\f{\hat{m}\log^5 w}{w}})$.  When $w = O(\log n)$ the cost of taking the transpose is cheaper
since sorting and merging a packed word takes unit time via table lookup.
Setting $h=\log n$, the total time is
\[
O\paren{K\log\log n + \f{\hat{m}}{K\log n} + K(\log \log n)^2}
\]
which is $O(\sqrt{\f{\hat{m}(\log\log n)^2}{\log n}})$ when $K = \sqrt{\f{\hat{m}}{\log n(\log\log n)^2}}$.


\appendix

\section{Removing the Bounded Degree Assumption}\label{sect:nodegreebound}

Invariants~\ref{inv:mass} and \ref{inv:superchunk} imply that there are $J = \Theta(\hat{m}/(Kh))$ superchunks with non-$\bottom$ IDs.
However, Invariant~\ref{inv:mass} cannot be satisfied (as stated) unless the maximum degree is bounded by $O(K)$.  
One way to guarantee this is to physically split up high degree vertices, replacing each $v$ with a cycle on new vertices 
$v_1,\ldots,v_{\ceil{\deg(v)/\Theta(K)}}$, each of which is responsible for $\Theta(K)$ of $v$'s edges.  This is the method
used by Frederickson~\cite{Frederickson85}, who actually demanded that the maximum degree be 3 at all times!  

This vertex-splitting can be effectively simulated in our algorithm as follows.  If $\deg(v)\ge K/2$, replace the principle copy of $v$
in its Euler tour with an interval of artificial principle vertices $v_1,\ldots,v_{\ceil{\deg(v)/(K/2)}}$, each of which is responsible
for between $K/2$ and $K$ of $v$'s edges.  Invariant~\ref{inv:mass} is therefore maintained w.r.t.~this modified tour.
To keep the mass of artificial vertices between $K/2$ and $K$, each edge insertion/deletion may require
splitting an artificial vertex or merging two consecutive artificial vertices.
When the Euler tour changes we always preserve the invariant that $v$'s artificial vertices form a contiguous interval in the tour.

\section{Dealing with Short Lists}\label{sect:shortlists}

Until now we have assumed for simplicity that all superchunks have proper IDs in $[J]$.
It is important that we {\em not} give out IDs to short lists (consisting of less than $h/2$ chunks) 
because the running time of the algorithm is linear in the maximum ID $J$.  
The modifications needed to deal with short lists are tedious but minor.

Consider an $\INSERTEDGE(u,v)$ operation where $u$ and $v$ are in lists $L_0,L_1$
and $L_1$ is a short list consisting of one superchunk $S$ with $\ID(S)=\bottom$.  If $L_0$ is not short
(or if it is short but the combined list $L_0L_1$ will not be short) then we retrieve an unused ID, say $i$,
set $\ID(S)\leftarrow i$, set $\ChADJ(i,i) \leftarrow \ChADJ_S$, and destroy $\ChADJ_S$.   By Invariant~\ref{inv:EulerTour},
$S$ was not incident to any other superchunk, so $\ChADJ(i,j)=0$ (the all-zero matrix) for all $j\neq i$. 
At this point $S$ violates Invariant~\ref{inv:superchunk} (it is too small), so we need to merge it with the
last superchunk in $L_0$ and resplit it along a different chunk boundary, in $O(J)$ time.

The modifications to $\DELETEEDGE(u,v)$ are analogous.  If we delete a tree edge $(u,v)$, splitting
its component into $T_0$ and $T_1$ with associated Euler tours $L_0$ and $L_1$, and
$\REPLACEMENTEDGE(u,v)$ fails to find an edge joining $L_0$ and $L_1$, we need to check
whether $L_0$ (and $L_1$) are short.  If so let $S$ be the superchunk in $L_0$.
We allocate and set $\ChADJ_S \leftarrow \ChADJ(\ID(S),\ID(S))$, 
then set $\ChADJ(\ID(S),\ID(S)) \leftarrow 0$ and finally retire $\ID(S)$.

The implementation of $\REPLACEMENTEDGE(L_0,L_1)$ is different if $L_0$ and $L_1$ were originally
in a short list $L = \Euler(T)$ before a tree edge in $T$ was deleted.  Suppose $L$ originally had 
one superchunk $S$, whose chunk adjacency was stored in $\ChADJ_S$.  
After $O(1)$ splits and joins, both $L_0$'s chunks and $L_1$'s chunks 
occupy $O(1)$ intervals of the rows and columns of $\ChADJ_S$.  Of course $\ChADJ_S$ is represented
as a list of its 1 positions in row-major order, so we can isolate the correct intervals of rows and columns
in $O(h^2\log^3 h/w)$ time.
If there is any 1 there, say at location $\ChADJ_S(k,l)$, then we know that there is an edge between $L_0$ and $L_1$,
and can find it in $O(K)$ time be examining chunks $k$ and $l$.  
The permutation of rows/columns in $\ChADJ_S$ must be updated to reflect any splits and joins that take place,
and if no replacement edge is discovered, $\ChADJ_S$ must be split into two lists representing matrices 
$\ChADJ_{S_0}$ and $\ChADJ_{S_1}$, to be identified with the single superchunks $S_0$ and $S_1$ in $L_0$ and $L_1$, respectively.

\newpage
\section{Summary of Prior Work}

\newcommand{\RANDLV}{{Randomized Las Vegas}}
\newcommand{\RANDMC}{{Randomized Monte Carlo}}
\begin{table}[h]
\scalebox{0.96}{\parbox{6.5in}{
\begin{tabular}{ll@{\istrut[3]{0}\hcm[.5]}l@{\hcm[.3]}l}
\multicolumn{4}{c}{\sc Worst Case Data Structures\istrut[3]{0}}\\
\multicolumn{1}{l}{\sc Ref.}	&	\multicolumn{1}{l}{\sc Update Time}	&\multicolumn{1}{@{\hcm[0]}l}{\sc Query Time}	& \multicolumn{1}{@{\hcm[0]}l}{\sc Notes}\\\hline
\cite{Frederickson85}	&	$O\paren{\sqrt{m}}$	& $O\paren{1}$		& \istrut{4.5}\\
\cite{EppsteinGIN97,Frederickson85}	&	$O\paren{\sqrt{n}}$	& $O\paren{1}$		& \cite{Frederickson85} + sparsification~\cite{EppsteinGIN97}.\istrut[6]{0}\\
\rb{-2}{\cite{KapronKM13}}	&	\rb{-2}{$O\paren{c\log^5 n}$}	& \rb{-2}{$O\paren{\log n/\log\log n}$} 	& \RANDMC;\\
\rb{-2}{\cite{GibbKKT15}}		&	\rb{-2}{$O\paren{c\log^4 n}$}	& \rb{-2}{$O\paren{\log n/\log\log n}$}			& \rb{1.5}{no connectivity witness;}\\
				&						&							&  \rb{3}{$n^c$ opers. err with prob.~$n^{-c}$.}\\
\rb{-2}{\bf new}		&	$O\paren{\sqrt{\f{n(\log\log n)^2}{\log n}}}$	& \rb{-2}{$O\paren{1}$}				& \rb{-2}{$w=\Omega(\log n)$}\\
				&	$O\paren{\sqrt{\f{n\log^5 w}{w}}}$			& \\\hline
&\\
\multicolumn{4}{c}{\sc Amortized Data Structures\istrut[3]{6}}\\
\multicolumn{1}{l}{\sc Ref.}	&	\multicolumn{1}{l}{\sc Amort.~Update}	& \multicolumn{1}{@{\hcm[0]}l}{\sc W.C.~Query}	& \multicolumn{1}{@{\hcm[0]}l}{\sc Notes}\\\hline
\cite{HenzingerK99}							&	$O\paren{\log^3 n}$	& $O\paren{\log n / \log\log n}$ 	& \RANDLV.\istrut{5.25}\\
\cite{HenzingerT97}							&	$O\paren{\log^2 n}$ 	& $O\paren{\log n/\log\log n}$ & \RANDLV.\\
\cite{HolmLT01}\					&	$O\paren{\log^2 n}$ 	& $O\paren{\log n/\log\log n}$ 	&\\
\cite{Thorup00}										&	$O\paren{\log n(\log\log n)^3}$ 	& $O\paren{\log n/\log\log\log n}$ & \RANDLV.\\
\cite{Wulff-Nilsen13}									&	$O\paren{\log^2 n/\log\log n}$		& $O\paren{\log n/\log\log n}$	\\\hline
&\\
\multicolumn{4}{c}{\sc Amort./Worst Case Lower Bounds\istrut[3]{6}}\\
\multicolumn{1}{l}{\sc Ref.}	&	\multicolumn{1}{l}{\sc Update Time $t_u$}	& \multicolumn{1}{@{\hcm[0]}l}{\sc Query Time $t_q$}	& \multicolumn{1}{@{\hcm[0]}l}{\sc Notes}\\\hline
\zero{\cite{FredmanS89,FredmanH98,MiltersenSVT94}}			&						& $t_q = \Omega\paren{\log n / \log(t_u\log n)}$		&\istrut[4]{5.25}\\
\cite{PatrascuD06}						&	$t_u = \Omega\paren{\log n / \log(t_q/t_u)}$ 	& $t_q = \Omega\paren{\log n / \log(t_u/t_q)}$	& Implies $\max\{t_u,t_q\} = \Omega(\log n)$.\istrut[4]{0}\\
\cite{PatrascuT11}							&	$o\paren{\log n}$	\hfill implies				& $\Omega\paren{n^{1-o(1)}}$\\\hline
\end{tabular}
}}
\caption{\label{table:priorwork}A survey of dynamic connectivity results.  
The lower bounds hold in the cell probe model with word size $w=\Theta(\log n)$.}
\end{table}

\end{document}